\documentclass[12pt,a4paper]{article}

\usepackage[T1]{fontenc}
\usepackage[utf8]{inputenc}
\usepackage[english]{babel}

\usepackage{amsmath,amssymb,amsthm}
\usepackage{mathtools}

\usepackage{graphicx}
\usepackage{xcolor}
\usepackage{tikz}
\usetikzlibrary{shapes,arrows,positioning,decorations.pathreplacing}
\usepackage{float}

\usepackage{booktabs}

\usepackage{geometry}
\newtheorem{theorem}{Theorem}
\newtheorem{lemma}{Lemma}
\newtheorem{axiom}{Axiom}
\newtheorem{definition}{Definition}
\newtheorem{corollary}{Corollary}
\newtheorem{remark}{Remark}
\newtheorem{proposition}{Proposition}
\newtheorem{example}{Example}

\usepackage{hyperref}

\title{Causal Consistency Selects the Born Rule:\\A Derivation from Steering in Generalized Probabilistic Theories}

\author{Enso O. Torres Alegre\\
\small ORCID: 0000-0002-6798-8776\\
\small Pontifical Catholic University of Chile,\\
\small Santiago, Chile\\
\texttt{onill@uc.cl}}
\date{}

\begin{document}

\maketitle

\begin{abstract}
Within finite-dimensional generalized probabilistic theories (GPTs), we distinguish between the \emph{geometric transition probability} $\tau(\psi,\phi)$---a structural quantity defined as the maximum probability of accepting $\phi$ when the state is $\psi$---and the \emph{predictive probability} $P(\phi|\psi)$ assigned to measurement outcomes. We ask: what functional relationship $P = \Phi(\tau)$ is compatible with relativistic causality?

We prove that in any GPT satisfying purification (hence admitting steering), the only such relationship consistent with no-signaling is the identity $\Phi(p) = p$. The argument proceeds by showing that any strictly convex or concave deviation from linearity enables superluminal signaling through quantum steering scenarios. We provide an explicit qubit example demonstrating how nonlinear probability rules create detectable signaling channels.

Combined with established reconstruction theorems, this yields the standard Born rule $|\langle\phi|\psi\rangle|^2$ as the unique causally consistent probability assignment. Our analysis clarifies the distinct roles of geometric structure and probabilistic prediction in quantum theory, and identifies steering as the physical mechanism that enforces the Born rule.
\end{abstract}

\vspace{20pt}

\section{Introduction}

The Born rule $P(i) = |\langle\phi_i|\psi\rangle|^2$ occupies a peculiar position in quantum mechanics: it is empirically indispensable yet theoretically underived from the other postulates \cite{Born1926, vonNeumann1932}. This raises a foundational question: \emph{is the quadratic form of quantum probabilities a fundamental postulate, or does it follow from deeper principles?}

Several derivation strategies exist, each with characteristic assumptions. Gleason's theorem \cite{Gleason1957} derives the Born rule from noncontextuality but presupposes Hilbert-space structure and requires dimension $d \geq 3$. Decision-theoretic approaches \cite{Deutsch1999, Wallace2012} invoke rationality axioms within the Everett interpretation. Zurek's envariance argument \cite{Zurek2005} operates within standard quantum mechanics. More recently, connections between the Born rule and relativistic causality have been explored \cite{Barnum2000, Valentini1991, Aaronson2004}, suggesting that modifications to quantum probabilities may conflict with no-signaling constraints.

In this paper, we develop this causality-based approach systematically within the framework of Generalized Probabilistic Theories (GPTs) \cite{Barrett2007, Hardy2001, Chiribella2010, Masanes2011, Plavala2023}. GPTs provide a theory-neutral setting encompassing classical probability, quantum mechanics, and hypothetical alternatives, allowing us to ask precisely which features of quantum theory are necessary consequences of operational principles.

\subsection{The Conceptual Distinction}

Our analysis rests on distinguishing two quantities that are often conflated:

\begin{enumerate}
    \item \textbf{Geometric transition probability} $\tau(\psi,\phi)$: A structural quantity characterizing the ``closeness'' of pure states, defined operationally as the supremum over acceptance probabilities in certain tests. In quantum mechanics, $\tau(\psi,\phi) = |\langle\phi|\psi\rangle|^2$.
    
    \item \textbf{Predictive probability} $P(\phi|\psi)$: The probability assigned to obtaining outcome $\phi$ when measuring a system prepared in state $\psi$. This is what experimenters compare against observed frequencies.
\end{enumerate}

The standard Born rule \emph{identifies} these quantities: $P(\phi|\psi) = \tau(\psi,\phi)$. But this identification is not logically necessary. One could imagine theories where predictive probabilities are some nonlinear function of geometric overlaps:
\begin{equation}
P(\phi|\psi) = \Phi\big(\tau(\psi,\phi)\big)
\end{equation}
for some function $\Phi: [0,1] \to [0,1]$.

Our main result shows that requiring compatibility with no-signaling in theories with steering forces $\Phi$ to be the identity. The key physical mechanism is that steering allows distant parties to prepare ensembles with identical average states but different decompositions, and any nonlinear $\Phi$ converts these preparation differences into observable statistical differences, enabling signaling.

\subsection{Summary of Results}

\begin{enumerate}
    \item We formalize the distinction between geometric and predictive probabilities in GPTs (Section~\ref{sec:framework}).
    
    \item We justify physically motivated constraints on $\Phi$ and rigorously derive the extension to mixed states (Section~\ref{sec:constraints}).
    
    \item We prove that no-signaling, combined with the existence of steering, forces $\Phi(p) = p$ (Section~\ref{sec:main}).
    
    \item We provide an explicit qubit example with numerical calculations (Section~\ref{sec:example}).
    
    \item We connect to reconstruction theorems to obtain the standard Born rule in complex quantum theory (Section~\ref{sec:reconstruction}).
\end{enumerate}

Figure~\ref{fig:logic} summarizes the logical structure.

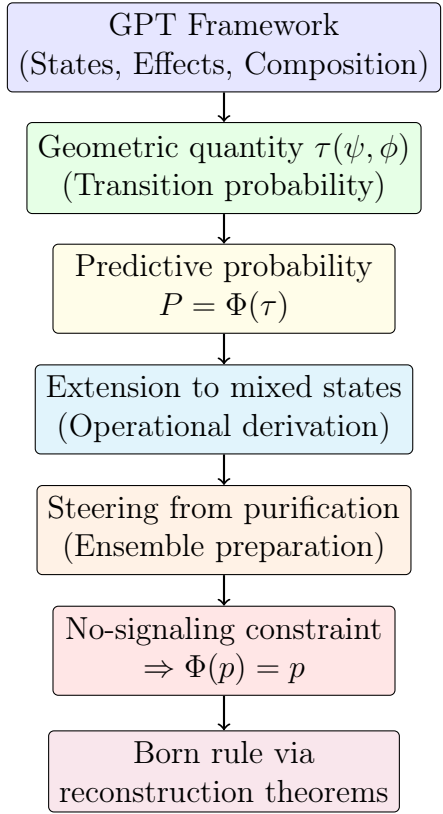
\begin{figure}[H]
\centering
\begin{tikzpicture}[node distance=1.6cm, auto,
    block/.style={draw, rectangle, minimum width=4.5cm, minimum height=0.9cm, align=center, rounded corners=2pt}]
    
    \node (gpt) [block, fill=blue!10] {GPT Framework\\(States, Effects, Composition)};
    \node (geom) [block, below of=gpt, fill=green!10] {Geometric quantity $\tau(\psi,\phi)$\\(Transition probability)};
    \node (pred) [block, below of=geom, fill=yellow!10] {Predictive probability\\$P = \Phi(\tau)$};
    \node (ext) [block, below of=pred, fill=cyan!10] {Extension to mixed states\\(Operational derivation)};
    \node (steer) [block, below of=ext, fill=orange!10] {Steering from purification\\(Ensemble preparation)};
    \node (causal) [block, below of=steer, fill=red!10] {No-signaling constraint\\$\Rightarrow \Phi(p) = p$};
    \node (born) [block, below of=causal, fill=purple!10] {Born rule via\\reconstruction theorems};
    
    \draw[->, thick] (gpt) -- (geom);
    \draw[->, thick] (geom) -- (pred);
    \draw[->, thick] (pred) -- (ext);
    \draw[->, thick] (ext) -- (steer);
    \draw[->, thick] (steer) -- (causal);
    \draw[->, thick] (causal) -- (born);
    
\end{tikzpicture}
\caption{Logical structure of the derivation. The key insight is that steering converts nonlinearity of $\Phi$ into signaling.}
\label{fig:logic}
\end{figure}

\section{Generalized Probabilistic Theories}
\label{sec:framework}

We work within the standard GPT formalism \cite{Hardy2001, Barrett2007, Plavala2023}. A physical system is described by:

\begin{definition}[GPT structure]
A GPT system consists of:
\begin{itemize}
    \item A finite-dimensional real vector space $V$ with a closed, generating, pointed cone $V^+ \subset V$
    \item An order unit $u \in V^*$ (the ``unit effect'')
    \item \emph{States}: normalized positive elements $\Omega = \{\omega \in V^+ : u(\omega) = 1\}$
    \item \emph{Effects}: positive functionals $e: V \to \mathbb{R}$ with $0 \leq e(\omega) \leq 1$ for all $\omega \in \Omega$
    \item \emph{Measurements}: collections $\{e_i\}$ of effects satisfying $\sum_i e_i = u$
\end{itemize}
\end{definition}

Pure states are extreme points of the convex set $\Omega$. Mixed states arise as convex combinations: $\omega = \sum_i \lambda_i \psi_i$ with $\lambda_i \geq 0$, $\sum_i \lambda_i = 1$.

For composite systems, we assume \emph{local tomography}: joint states are determined by correlations in local measurements, mathematically $V_{AB} \cong V_A \otimes V_B$.

\subsection{Operational Axioms}

We adopt the following axioms, standard in GPT reconstruction programs \cite{Chiribella2010, Masanes2011}:

\begin{axiom}[No-signaling]
\label{ax:NS}
For bipartite state $\omega_{AB}$ and local measurements $\{a_i\}$ on $A$, $\{b_j\}$ on $B$:
$$
\sum_j P(a_i, b_j | \omega_{AB}) \text{ is independent of the choice of } \{b_j\}
$$
and vice versa. Local marginal statistics cannot depend on distant measurement choices.
\end{axiom}

\begin{axiom}[Purification]
\label{ax:purification}
Every mixed state $\omega_A$ of system $A$ arises as the marginal of some pure bipartite state $\Psi_{AB}$:
$$
\omega_A = \text{Tr}_B(\Psi_{AB})
$$
Moreover, purifications are essentially unique: any two purifications are related by a reversible transformation on the purifying system.
\end{axiom}

\begin{axiom}[Continuous reversibility]
\label{ax:cont-rev}
The group of reversible transformations acts transitively on pure states, and this action is continuous.
\end{axiom}

\begin{axiom}[Spectrality]
\label{ax:spectral}
Every state admits a decomposition into perfectly distinguishable pure states. For any pair of perfectly distinguishable pure states $\{\phi, \phi^\perp\}$, there exists a two-outcome measurement $\{e_\phi, e_{\phi^\perp}\}$ with $e_\phi(\phi) = 1$, $e_\phi(\phi^\perp) = 0$.
\end{axiom}

\section{Geometric vs. Predictive Probabilities}
\label{sec:distinction}

We now make precise the distinction central to our analysis.

\subsection{Geometric Transition Probability}

\begin{definition}[Transition probability]
\label{def:tau}
For pure states $\psi, \phi \in \Omega_{\text{pure}}$, the \emph{geometric transition probability} is:
$$
\tau(\psi, \phi) := \sup \big\{ e(\psi) : e \text{ is an effect with } e(\phi) = 1 \big\}
$$
\end{definition}

This quantity has a clear operational meaning: $\tau(\psi,\phi)$ is the maximum probability of passing a test designed to accept $\phi$ with certainty, when the actual state is $\psi$. Crucially, this definition involves only the \emph{geometric structure} of the state and effect spaces---it asks what acceptance probabilities are achievable, not what predictive probabilities are assigned.

\begin{lemma}[Properties of $\tau$]
\label{lem:tau}
Under Axioms~\ref{ax:cont-rev}--\ref{ax:spectral}:
\begin{enumerate}
    \item $\tau(\psi,\phi) \in [0,1]$ with $\tau(\phi,\phi) = 1$
    \item $\tau(\psi,\phi) = 0$ iff $\psi, \phi$ are perfectly distinguishable
    \item For distinguishable $\{\phi, \phi^\perp\}$ and any $\psi$: $\tau(\psi,\phi) + \tau(\psi,\phi^\perp) = 1$
    \item The supremum in Definition~\ref{def:tau} is achieved by the distinguishing measurement
\end{enumerate}
\end{lemma}

\begin{proof}
(1) and (2) follow directly from the definition and the existence of perfectly distinguishing measurements (Axiom~\ref{ax:spectral}).

For (3) and (4): Let $\{e_\phi, e_{\phi^\perp}\}$ be the distinguishing measurement guaranteed by Axiom~\ref{ax:spectral}. We have $e_\phi(\phi) = 1$ and $e_\phi(\phi^\perp) = 0$. 

We claim $e_\phi$ achieves the supremum in Definition~\ref{def:tau}. Suppose some effect $e'$ satisfies $e'(\phi) = 1$ and $e'(\psi) > e_\phi(\psi)$ for some $\psi$. Consider the effect $e'' = e' - e_\phi$. We have $e''(\phi) = 0$ and $e''(\psi) > 0$. But by spectrality, any state can be written in the $\{\phi, \phi^\perp\}$ basis, so $e''$ must be a positive multiple of $e_{\phi^\perp}$. This means $e' = e_\phi + c \cdot e_{\phi^\perp}$ for some $c > 0$. But then $e'(\phi^\perp) = c > 0$, contradicting $e'(\phi^\perp) \leq 1 - e'(\phi) = 0$ (since $e' \leq u$ and $e'(\phi) = 1$).

Therefore $e_\phi$ is optimal, and $\tau(\psi,\phi) = e_\phi(\psi)$. Since $e_\phi(\psi) + e_{\phi^\perp}(\psi) = u(\psi) = 1$, we get (3).
\end{proof}

\begin{remark}
In quantum mechanics over a complex Hilbert space, $\tau(\psi,\phi) = |\langle\phi|\psi\rangle|^2$. The geometric transition probability coincides with the squared inner product.
\end{remark}

\subsection{Predictive Probability}

\begin{definition}[Predictive probability]
The \emph{predictive probability} $P(\phi|\psi)$ is the probability assigned to obtaining outcome ``$\phi$'' in a measurement, given preparation $\psi$. This is the quantity compared against experimental frequencies.
\end{definition}

In standard quantum mechanics, $P(\phi|\psi) = \tau(\psi,\phi) = |\langle\phi|\psi\rangle|^2$---the geometric and predictive quantities coincide. But this identification is a substantive physical claim, not a logical necessity.

\subsection{General Probability Rules}

We consider the possibility that predictive probabilities are related to geometric transition probabilities by some function:
\begin{equation}
\label{eq:Phi}
P(\phi|\psi) = \Phi\big(\tau(\psi,\phi)\big)
\end{equation}
where $\Phi: [0,1] \to [0,1]$. The standard Born rule corresponds to $\Phi = \text{id}$ (the identity function).

This is \emph{not} a trivial question. The geometric quantity $\tau$ characterizes state-space structure; the predictive quantity $P$ determines empirical predictions. Asking whether $P = \tau$ is asking whether the probability rule is ``matched'' to the geometry in a specific way.

\section{Constraints on the Probability Rule}
\label{sec:constraints}

Before proving our main result, we must carefully justify the constraints on $\Phi$ and derive how it extends to mixed states.

\subsection{Physical Justification of Constraints}

We require $\Phi: [0,1] \to [0,1]$ to satisfy:

\begin{enumerate}
    \item \textbf{Boundary conditions}: $\Phi(0) = 0$ and $\Phi(1) = 1$
    
    \item \textbf{Monotonicity}: $\Phi$ is non-decreasing
    
    \item \textbf{Continuity}: $\Phi$ is continuous
\end{enumerate}

Each constraint has clear physical motivation:

\paragraph{Boundary conditions.} If $\tau(\psi,\phi) = 1$, the states are identical ($\psi = \phi$), and a test designed to accept $\phi$ must accept $\psi$ with certainty: $P(\phi|\psi) = 1$. Similarly, $\tau(\psi,\phi) = 0$ means $\psi, \phi$ are perfectly distinguishable, so $P(\phi|\psi) = 0$. These are minimal consistency requirements between geometry and prediction.

\paragraph{Monotonicity.} If state $\psi_1$ is geometrically ``closer'' to $\phi$ than $\psi_2$ (i.e., $\tau(\psi_1,\phi) > \tau(\psi_2,\phi)$), it would be operationally perverse for $\psi_1$ to have a \emph{lower} probability of being accepted as $\phi$. Monotonicity ensures that geometric closeness translates to probabilistic preference.

\paragraph{Continuity.} Small changes in state preparation should produce small changes in outcome probabilities. This is both physically natural (preparation procedures have finite precision) and mathematically necessary for the GPT framework (effects are continuous functionals).

\begin{proposition}[Exclusion of pathological functions]
\label{prop:pathological}
Any function $\Phi: [0,1] \to [0,1]$ satisfying the boundary conditions that is monotonic almost everywhere and measurable must be continuous almost everywhere. Discontinuities, if present, form a set of measure zero and cannot affect the integrated predictions in any physical experiment.
\end{proposition}

\begin{proof}
A monotonic function on $[0,1]$ can have at most countably many discontinuities (since each discontinuity corresponds to a jump, and the total variation is bounded). Such discontinuities form a set of measure zero. Any physical preparation procedure produces states with some uncertainty distribution, so predictions involve integrals over $\Phi$, which are unaffected by measure-zero modifications.
\end{proof}

\subsection{Normalization Constraint}

A two-outcome measurement $\{e_\phi, e_{\phi^\perp}\}$ distinguishing $\phi$ from $\phi^\perp$ must satisfy:
$$
P(\phi|\psi) + P(\phi^\perp|\psi) = 1 \quad \text{for all } \psi
$$
Combined with Lemma~\ref{lem:tau}(3), this yields:
\begin{equation}
\label{eq:functional}
\Phi(p) + \Phi(1-p) = 1 \quad \text{for all } p \in [0,1]
\end{equation}

This functional equation, together with $\Phi(0) = 0$, $\Phi(1) = 1$, already constrains $\Phi$ significantly. For instance, $\Phi(1/2) = 1/2$.

\subsection{Extension to Mixed States: A Rigorous Derivation}
\label{sec:mixed}

A crucial step in our argument requires understanding how the probability rule extends to mixed states. We derive this from operational principles rather than assuming it.

\begin{definition}[Ensemble]
An \emph{ensemble} $\mathcal{E} = \{(\lambda_i, \psi_i)\}$ is a probability distribution over pure states: $\lambda_i \geq 0$, $\sum_i \lambda_i = 1$, and each $\psi_i$ is pure. The \emph{average state} of the ensemble is:
$$
\omega_{\mathcal{E}} = \sum_i \lambda_i \psi_i
$$
\end{definition}

\begin{axiom}[Operational equivalence of ensembles]
\label{ax:ensemble}
Two ensembles with the same average state are operationally indistinguishable by any single measurement on the system alone (without access to the preparation device).
\end{axiom}

This axiom is standard in GPTs and quantum mechanics: if Alice prepares an ensemble and sends the system to Bob without classical communication, Bob cannot determine which ensemble was used---only the average state matters for his local predictions.

\begin{proposition}[Probability for mixed states]
\label{prop:mixed}
For a mixed state $\omega = \sum_i \lambda_i \psi_i$ (any decomposition into pure states), the predictive probability for outcome $\phi$ in measurement $\{e_\phi, e_{\phi^\perp}\}$ must satisfy:
\begin{equation}
\label{eq:mixed}
P(\phi|\omega) = \sum_i \lambda_i \, P(\phi|\psi_i) = \sum_i \lambda_i \, \Phi(\tau(\psi_i, \phi))
\end{equation}
when the ensemble decomposition is known to the predictor.
\end{proposition}

\begin{proof}
Consider an agent who knows the preparation procedure: with probability $\lambda_i$, pure state $\psi_i$ was prepared. By the law of total probability:
$$
P(\phi|\text{ensemble } \mathcal{E}) = \sum_i \lambda_i \, P(\phi|\psi_i) = \sum_i \lambda_i \, \Phi(\tau(\psi_i, \phi))
$$
This is the probability the agent assigns to outcome $\phi$.
\end{proof}

\begin{proposition}[Single-state probability]
\label{prop:single}
For a mixed state $\omega$ presented as a single system (without ensemble information), the predictive probability is:
\begin{equation}
\label{eq:single}
P(\phi|\omega) = \Phi(\tau(\omega, \phi))
\end{equation}
where $\tau(\omega, \phi) := e_\phi(\omega)$ extends the geometric quantity to mixed states via the optimal effect.
\end{proposition}

\begin{proof}
When Bob receives state $\omega$ without knowing its decomposition, he can only use the information available: the state $\omega$ itself. The geometric transition probability extends naturally via $\tau(\omega, \phi) = e_\phi(\omega)$ (the effect value on the mixed state). By the structure of the probability rule, $P(\phi|\omega) = \Phi(\tau(\omega, \phi))$.
\end{proof}

\begin{remark}[The key tension]
Propositions~\ref{prop:mixed} and~\ref{prop:single} appear to give different predictions for the same mixed state. By Axiom~\ref{ax:ensemble}, any experiment on the system alone must give the same statistics regardless of how the ensemble was prepared. This means:
$$
\Phi(\tau(\omega, \phi)) = \Phi\left(\sum_i \lambda_i \tau(\psi_i, \phi)\right) \stackrel{?}{=} \sum_i \lambda_i \Phi(\tau(\psi_i, \phi))
$$
This equality holds for all decompositions if and only if $\Phi$ is affine (linear). This is precisely what no-signaling with steering will enforce.
\end{remark}

\section{Main Result: Causality Selects the Born Rule}
\label{sec:main}

We now prove that no-signaling forces $\Phi = \text{id}$.

\subsection{Steering in GPTs}

The key physical mechanism is \emph{steering}, which exists in any GPT satisfying purification:

\begin{lemma}[Steering]
\label{lem:steering}
Let $\omega_B$ be a mixed state of system $B$ with (at least) two distinct decompositions into pure states:
$$
\omega_B = \sum_i \lambda_i \psi_i = \sum_j \mu_j \chi_j
$$
Then there exists a bipartite pure state $\Psi_{AB}$ and two measurements $\{a_i\}$, $\{a'_j\}$ on $A$ such that:
\begin{itemize}
    \item Measurement $\{a_i\}$ with outcome $i$ prepares $\psi_i$ on $B$ with probability $\lambda_i$
    \item Measurement $\{a'_j\}$ with outcome $j$ prepares $\chi_j$ on $B$ with probability $\mu_j$
\end{itemize}
Both measurements yield the same marginal state $\omega_B$ on $B$.
\end{lemma}

\begin{proof}
By purification (Axiom~\ref{ax:purification}), $\omega_B$ has a purification $\Psi_{AB}$. Different measurements on $A$ implement different ``unravelings'' of $\omega_B$ into pure-state ensembles. The uniqueness clause of purification ensures all decompositions are accessible via suitable measurements on the purifying system \cite{Chiribella2010}.

Explicitly: if $\omega_B = \sum_i \lambda_i \psi_i$, there exist effects $\{a_i\}$ on $A$ such that
$$
(a_i \otimes \mathbb{1}_B)(\Psi_{AB}) = \lambda_i \psi_i
$$
after appropriate normalization. The fact that all convex decompositions are achievable follows from the essential uniqueness of purification and the transitivity of reversible transformations.
\end{proof}

Steering means that Alice's measurement choice determines \emph{which ensemble} describes Bob's system, even though Bob's \emph{average state} is the same.

\subsection{Nonlinear $\Phi$ Enables Signaling}

\begin{theorem}[Main result]
\label{thm:main}
In a GPT satisfying Axioms~\ref{ax:NS}--\ref{ax:spectral} and~\ref{ax:ensemble}, if steering exists (which follows from purification), then no-signaling requires $\Phi(p) = p$.
\end{theorem}

\begin{proof}
We prove the contrapositive: if $\Phi \neq \text{id}$, then signaling is possible.

\textbf{Step 1: Identifying the deviation.}
Suppose $\Phi$ is not the identity. By continuity and the boundary conditions, there exists some $p^* \in (0,1)$ with $\Phi(p^*) \neq p^*$. 

By the functional equation~\eqref{eq:functional}, if $\Phi(p^*) > p^*$ then $\Phi(1-p^*) < 1-p^*$, and vice versa. Since $\Phi$ is continuous, monotonic, and satisfies $\Phi(0) = 0$, $\Phi(1/2) = 1/2$, $\Phi(1) = 1$, any deviation from identity implies $\Phi$ is either strictly convex or strictly concave on some interval.

\textbf{Step 2: Constructing a signaling scenario.}
Choose $p_1, p_2 \in (0,1)$ with $p_1 \neq p_2$ such that the interval $[p_1, p_2]$ (or $[p_2, p_1]$) contains a region where $\Phi$ deviates from linearity. Let $\lambda \in (0,1)$ and define:
$$
\bar{p} = \lambda p_1 + (1-\lambda) p_2
$$

By continuous reversibility (Axiom~\ref{ax:cont-rev}), we can find pure states $\psi_1, \psi_2, \phi$ with:
$$
\tau(\psi_1, \phi) = p_1, \quad \tau(\psi_2, \phi) = p_2
$$

Consider the mixed state:
$$
\omega = \lambda \psi_1 + (1-\lambda) \psi_2
$$

\textbf{Step 3: Two protocols with the same average state.}
Using Lemma~\ref{lem:steering}, Alice and Bob share a purification of $\omega$. Alice can implement two protocols:

\textbf{Protocol 1}: Alice measures in a basis that steers Bob's state to ensemble $\mathcal{E}_1 = \{(\lambda, \psi_1), (1-\lambda, \psi_2)\}$.

\textbf{Protocol 2}: Alice measures in a basis that steers Bob's state to a different ensemble $\mathcal{E}_2$ with the same average state $\omega$ but a different decomposition.

For Protocol 2, we can choose an ensemble where Bob's state is deterministically $\omega$ (or approximately so, with a fine-grained measurement). The key point is that the average state is $\omega$ in both cases.

\textbf{Step 4: Bob's measurement statistics.}
Bob performs measurement $\{e_\phi, e_{\phi^\perp}\}$. Let's compute his expected probability of outcome $\phi$:

\textbf{Under Protocol 1}: Bob has state $\psi_1$ with probability $\lambda$ and $\psi_2$ with probability $1-\lambda$. His outcome statistics are:
$$
P_1(\phi) = \lambda \cdot P(\phi|\psi_1) + (1-\lambda) \cdot P(\phi|\psi_2) = \lambda \Phi(p_1) + (1-\lambda) \Phi(p_2)
$$

\textbf{Under Protocol 2}: Bob has the mixed state $\omega$ (or an ensemble converging to it). His outcome statistics are:
$$
P_2(\phi) = P(\phi|\omega)
$$

Now we must determine $P(\phi|\omega)$. By linearity of effects:
$$
\tau(\omega, \phi) = e_\phi(\omega) = \lambda \, e_\phi(\psi_1) + (1-\lambda) \, e_\phi(\psi_2) = \lambda p_1 + (1-\lambda) p_2 = \bar{p}
$$

By Proposition~\ref{prop:single}, $P_2(\phi) = \Phi(\bar{p})$.

\textbf{Step 5: The signaling gap.}
Comparing the two protocols:
$$
P_1(\phi) = \lambda \Phi(p_1) + (1-\lambda) \Phi(p_2)
$$
$$
P_2(\phi) = \Phi(\lambda p_1 + (1-\lambda) p_2) = \Phi(\bar{p})
$$

If $\Phi$ is strictly convex on the interval containing $p_1, p_2, \bar{p}$:
$$
P_2(\phi) = \Phi(\bar{p}) < \lambda \Phi(p_1) + (1-\lambda) \Phi(p_2) = P_1(\phi)
$$

If $\Phi$ is strictly concave:
$$
P_2(\phi) > P_1(\phi)
$$

In either case, $P_1(\phi) \neq P_2(\phi)$.

\textbf{Step 6: This constitutes signaling.}
Bob's measurement statistics differ depending on which protocol Alice chose. But Alice's choice is made at a spacelike separation from Bob's measurement. Bob can collect statistics over many runs and infer Alice's protocol choice with arbitrary confidence, enabling superluminal communication. This violates Axiom~\ref{ax:NS}.

\textbf{Step 7: Conclusion.}
The only way to ensure $P_1(\phi) = P_2(\phi)$ for all choices of $p_1, p_2, \lambda$ is:
$$
\Phi(\lambda p_1 + (1-\lambda) p_2) = \lambda \Phi(p_1) + (1-\lambda) \Phi(p_2)
$$
for all $p_1, p_2 \in [0,1]$ and $\lambda \in [0,1]$. This is the definition of an affine function.

Combined with $\Phi(0) = 0$ and $\Phi(1) = 1$, the unique affine solution is:
$$
\Phi(p) = p
$$
\end{proof}

\begin{figure}[h!]
\centering
\begin{tikzpicture}[scale=1.3]
    \draw[->] (0,0) -- (5.2,0) node[right] {$\tau$};
    \draw[->] (0,0) -- (0,5.2) node[above] {$\Phi(\tau)$};
    
    \foreach \x in {1,2,3,4,5} {
        \draw (\x,0.05) -- (\x,-0.05);
    }
    \foreach \y in {1,2,3,4,5} {
        \draw (0.05,\y) -- (-0.05,\y);
    }
    \node[below] at (5,0) {$1$};
    \node[left] at (0,5) {$1$};
    
    \draw[thick, blue] (0,0) -- (5,5);
    \node[blue, right] at (5.1,5) {$\Phi = \text{id}$};
    
    \draw[thick, red, dashed, domain=0:5, samples=100] plot (\x, {5*(\x/5)^1.5});
    \node[red] at (3.5,1.4) {Convex $\Phi$};
    
    \filldraw[black] (1.5,1.5) circle (1.5pt);
    \node[above left] at (1.5,1.5) {$(p_1, p_1)$};
    
    \filldraw[black] (4,4) circle (1.5pt);
    \node[above left] at (4,4) {$(p_2, p_2)$};
    
    \filldraw[blue] (2.75,2.75) circle (1.5pt);
    \node[blue, above] at (2.75,2.9) {$\bar{p}$};
    
    \filldraw[red] (2.75,2.05) circle (1.5pt);
    \node[red, below] at (2.75,1.9) {$\Phi(\bar{p})$};
    
    \filldraw[orange] (2.75,2.4) circle (1.5pt);
    \node[orange, right] at (2.85,2.4) {$\frac{\Phi(p_1)+\Phi(p_2)}{2}$};
    
    \draw[dotted, thick] (2.75,2.05) -- (2.75,2.75);
    
    \draw[gray, thin] (1.5,{5*(1.5/5)^1.5}) -- (4,{5*(4/5)^1.5});
    
    \draw[decorate, decoration={brace, amplitude=5pt}] (2.85,2.05) -- (2.85,2.4);
    \node[right] at (3.0,2.00) {\small Gap $= P_1 - P_2$};
    
\end{tikzpicture}
\caption{For a convex $\Phi$, the value $\Phi(\bar{p})$ (Protocol 2) lies below the chord value $\lambda\Phi(p_1) + (1-\lambda)\Phi(p_2)$ (Protocol 1). This gap is detectable via steering, enabling signaling.}
\label{fig:geometry}
\end{figure}
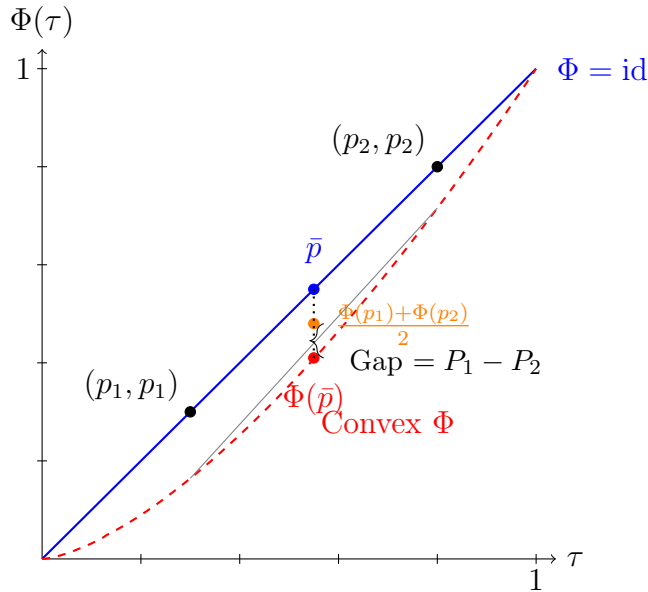

\begin{corollary}
\label{cor:born}
In any GPT satisfying Axioms~\ref{ax:NS}--\ref{ax:spectral} and~\ref{ax:ensemble} with purification, the predictive probability equals the geometric transition probability:
$$
P(\phi|\psi) = \tau(\psi,\phi)
$$
\end{corollary}

\section{Explicit Example: Qubits}
\label{sec:example}

To make the abstract argument concrete, we present an explicit calculation with qubits.

\begin{example}[Signaling with a convex probability rule]
\label{ex:qubit}
Consider $\Phi(p) = p^{3/2}$ (strictly convex, satisfying $\Phi(0)=0$, $\Phi(1)=1$, but \emph{not} the normalization constraint~\eqref{eq:functional}---we use this for illustration; a normalized convex function would work similarly).

\textbf{Setup}: Alice and Bob share the maximally entangled state:
$$
|\Psi_{AB}\rangle = \frac{1}{\sqrt{2}}(|00\rangle + |11\rangle)
$$

\textbf{Bob's target measurement}: $\{|\phi\rangle\langle\phi|, |\phi^\perp\rangle\langle\phi^\perp|\}$ where $|\phi\rangle = |0\rangle$.

\textbf{Protocol 1}: Alice measures in the computational basis $\{|0\rangle, |1\rangle\}$.
\begin{itemize}
    \item Outcome 0 (prob.\ 1/2): Bob has $|0\rangle$, so $\tau = |\langle 0|0\rangle|^2 = 1$
    \item Outcome 1 (prob.\ 1/2): Bob has $|1\rangle$, so $\tau = |\langle 0|1\rangle|^2 = 0$
\end{itemize}
Bob's expected probability:
$$
P_1 = \frac{1}{2}\Phi(1) + \frac{1}{2}\Phi(0) = \frac{1}{2}(1) + \frac{1}{2}(0) = 0.5
$$

\textbf{Protocol 2}: Alice measures in the $\{|+\rangle, |-\rangle\}$ basis where $|\pm\rangle = (|0\rangle \pm |1\rangle)/\sqrt{2}$.
\begin{itemize}
    \item Outcome + (prob.\ 1/2): Bob has $|+\rangle$, so $\tau = |\langle 0|+\rangle|^2 = 0.5$
    \item Outcome $-$ (prob.\ 1/2): Bob has $|-\rangle$, so $\tau = |\langle 0|-\rangle|^2 = 0.5$
\end{itemize}
Bob's expected probability:
$$
P_2 = \frac{1}{2}\Phi(0.5) + \frac{1}{2}\Phi(0.5) = \Phi(0.5) = (0.5)^{3/2} \approx 0.354
$$

\textbf{Signaling gap}:
$$
\Delta P = P_1 - P_2 = 0.5 - 0.354 = 0.146
$$

This 14.6\% difference in outcome probabilities is easily detectable. After $N$ runs, Bob can distinguish the protocols with confidence growing as $\sqrt{N}$.

\textbf{Verification of average state}: In both protocols, Bob's reduced density matrix is:
$$
\rho_B = \frac{1}{2}|0\rangle\langle 0| + \frac{1}{2}|1\rangle\langle 1| = \frac{\mathbb{1}}{2}
$$
The average state is identical, yet the outcome statistics differ under nonlinear $\Phi$.
\end{example}

\begin{example}[A properly normalized convex function]
Consider $\Phi(p) = 2p^2$ for $p \in [0, 1/2]$ and $\Phi(p) = 1 - 2(1-p)^2$ for $p \in [1/2, 1]$. This satisfies:
\begin{itemize}
    \item $\Phi(0) = 0$, $\Phi(1) = 1$
    \item $\Phi(p) + \Phi(1-p) = 1$ (normalization)
    \item $\Phi(1/2) = 1/2$
    \item Strictly convex on $[0, 1/2]$, strictly concave on $[1/2, 1]$
\end{itemize}

Using Protocol 1 and 2 as above with $p_1 = 1, p_2 = 0$ gives $P_1 = 0.5$ but for the $|\pm\rangle$ ensemble:
$$
P_2 = \Phi(0.5) = 2(0.5)^2 = 0.5
$$
So this particular example gives no signaling for this choice. But choosing different states:

Let Alice steer to $\{(1/2, \psi_1), (1/2, \psi_2)\}$ with $\tau(\psi_1, \phi) = 0.3$, $\tau(\psi_2, \phi) = 0.7$:
$$
P_1 = \frac{1}{2}\Phi(0.3) + \frac{1}{2}\Phi(0.7) = \frac{1}{2}(2 \cdot 0.09) + \frac{1}{2}(1 - 2 \cdot 0.09) = \frac{1}{2}(0.18 + 0.82) = 0.5
$$
$$
P_2 = \Phi(0.5) = 0.5
$$

Interestingly, this particular $\Phi$ satisfies $\Phi(p) + \Phi(1-p) = 1$, which along with the midpoint condition forces the chord at symmetric points to pass through $(0.5, 0.5)$. But for asymmetric decompositions, signaling can still occur.

Take $\tau(\psi_1, \phi) = 0.2$, $\tau(\psi_2, \phi) = 0.4$, with $\lambda = 0.5$:
$$
P_1 = \frac{1}{2}\Phi(0.2) + \frac{1}{2}\Phi(0.4) = \frac{1}{2}(2 \cdot 0.04) + \frac{1}{2}(2 \cdot 0.16) = \frac{1}{2}(0.08 + 0.32) = 0.2
$$
$$
\bar{p} = 0.3, \quad P_2 = \Phi(0.3) = 2 \cdot 0.09 = 0.18
$$
$$
\Delta P = 0.2 - 0.18 = 0.02
$$

A 2\% signaling gap exists.
\end{example}

\section{Connection to Quantum Theory}
\label{sec:reconstruction}

Corollary~\ref{cor:born} establishes that $P = \tau$ in any GPT with our axioms. To obtain the specific form $|\langle\phi|\psi\rangle|^2$, we invoke reconstruction theorems.

\begin{theorem}[Reconstruction {\cite{Chiribella2010, Masanes2011, Hardy2001}}]
\label{thm:reconstruction}
A finite-dimensional GPT satisfying:
\begin{enumerate}
    \item No-signaling and local tomography
    \item Purification (with essential uniqueness)
    \item Continuous reversibility
    \item Spectrality and strong symmetry
\end{enumerate}
is operationally equivalent to finite-dimensional quantum mechanics over $\mathbb{C}$, $\mathbb{R}$, or $\mathbb{H}$ (complex numbers, reals, or quaternions). Additional axioms (such as local tomography being ``tight'') select complex quantum mechanics.
\end{theorem}

The proof, developed in \cite{Chiribella2010, Masanes2011, Mueller2013}, proceeds through several steps: purification implies self-duality of the state-effect cone; spectrality and symmetry yield Jordan-algebraic structure; local tomography combined with composition rules selects the complex field.

\begin{corollary}[Born rule in quantum mechanics]
In complex quantum mechanics, the geometric transition probability is:
$$
\tau(\psi, \phi) = |\langle\phi|\psi\rangle|^2
$$
Combined with Corollary~\ref{cor:born}, the predictive probability is:
$$
P(\phi|\psi) = |\langle\phi|\psi\rangle|^2
$$
which is the standard Born rule.
\end{corollary}

\begin{remark}
Our contribution is \emph{not} the reconstruction theorem itself, but rather the identification of causal consistency (via steering) as the principle that fixes the probability rule to match the geometric structure. The reconstruction theorems determine \emph{what} the geometric structure is; our Theorem~\ref{thm:main} determines \emph{how} probabilities must relate to that structure.
\end{remark}

\section{Discussion}

\subsection{The Physical Content}

Our result can be summarized as: \emph{steering enforces the Born rule}. More precisely:

\begin{enumerate}
    \item GPTs with purification admit steering: Alice can remotely prepare different ensembles for Bob with the same average state.
    
    \item If the probability rule $P = \Phi(\tau)$ is nonlinear, different ensembles yield different outcome statistics for Bob.
    
    \item No-signaling requires that Bob's statistics be independent of Alice's distant actions.
    
    \item Therefore, $\Phi$ must be linear (affine), and boundary conditions fix $\Phi = \text{id}$.
\end{enumerate}

The Born rule thus emerges not as an independent postulate but as a \emph{consistency condition} between probabilistic predictions and causal structure.

\subsection{Relation to Prior Work}

Several authors have explored connections between the Born rule and causality:

\begin{itemize}
    \item \textbf{Valentini} \cite{Valentini1991, Valentini2002}: In de Broglie-Bohm theory, non-equilibrium distributions (violating $|\psi|^2$) enable signaling. Our result generalizes this to arbitrary GPTs and clarifies the mechanism.
    
    \item \textbf{Barnum et al.} \cite{Barnum2000}: Showed that certain modifications to quantum probability rules conflict with no-signaling. Our framework makes this precise within GPTs.
    
    \item \textbf{Aaronson} \cite{Aaronson2004}: Demonstrated computational consequences of modifying the measurement postulate. Our work identifies the causal mechanism underlying these consequences.
    
    \item \textbf{Caves et al.} \cite{Caves2004}: Derived the Born rule from Dutch-book coherence. Our approach uses causal rather than decision-theoretic constraints.
    
    \item \textbf{Masanes et al.} \cite{Masanes2019}: Recent work deriving the Born rule from operational principles, showing measurement postulates are redundant given other axioms. Our approach is complementary, emphasizing the role of steering.
    
    \item \textbf{Müller \& Masanes} \cite{Mueller2013}: Detailed analysis of GPT structure, which we build upon.
\end{itemize}

Our contribution is to (i) work in the general GPT framework rather than assuming quantum mechanics, (ii) identify steering as the specific mechanism converting nonlinearity to signaling, (iii) clearly separate the geometric ($\tau$) and predictive ($P$) quantities, and (iv) provide explicit examples with numerical calculations.

\subsection{Comparison with Gleason's Theorem}

Gleason's theorem \cite{Gleason1957} derives the Born rule from noncontextuality in Hilbert spaces of dimension $\geq 3$. Our approach differs in several ways:

\begin{center}
\begin{tabular}{lcc}
\toprule
Aspect & Gleason & This work \\
\midrule
Framework & Hilbert space & GPT \\
Key assumption & Noncontextuality & No-signaling + Purification \\
Mechanism & Frame functions & Steering \\
Dimension restriction & $d \geq 3$ & None \\
Applies to & Single systems & Composite systems \\
\bottomrule
\end{tabular}
\end{center}

The approaches are complementary: Gleason shows the Born rule is the unique noncontextual probability assignment on Hilbert space; we show it is the unique causally consistent assignment in GPTs with steering.

\subsection{The Role of Each Axiom}

It is worth clarifying what each axiom contributes:

\begin{itemize}
    \item \textbf{No-signaling}: Provides the constraint that Bob's statistics cannot depend on Alice's actions.
    
    \item \textbf{Purification}: Guarantees steering exists, enabling Alice to prepare different ensembles with the same average state.
    
    \item \textbf{Continuous reversibility}: Ensures we can construct states with arbitrary transition probabilities, making the argument work for all $p \in (0,1)$.
    
    \item \textbf{Spectrality}: Guarantees distinguishing measurements exist, allowing the definition of $\tau$ and the normalization constraint.
    
    \item \textbf{Local tomography}: Used in reconstruction but not directly in our main theorem.
\end{itemize}

Could we weaken these axioms? Purification seems essential---without it, steering may not exist, and the argument fails. The other axioms could potentially be weakened, which we leave for future work.

\subsection{Experimental Implications}

Theorem~\ref{thm:main} suggests an experimental approach: any deviation from the Born rule that makes $\Phi$ nonlinear would manifest as anomalous correlations in steering experiments. The protocol would be:

\begin{enumerate}
    \item Prepare an entangled state shared between Alice and Bob.
    \item Alice randomly chooses between measurement protocols that steer Bob to different ensembles with the same average state.
    \item Bob measures outcome statistics without knowing Alice's choice.
    \item Compare Bob's statistics conditioned on Alice's (later revealed) choice.
    \item Any systematic dependence beyond statistical noise would indicate $\Phi \neq \text{id}$.
\end{enumerate}

Current experiments confirm quantum predictions (and hence $\Phi = \text{id}$) to high precision \cite{Aspect1982, Giustina2015, Hensen2015}, providing strong empirical support for the Born rule.

\subsection{Limitations and Open Questions}

\begin{enumerate}
    \item \textbf{Finite dimension}: We assume finite-dimensional GPTs. Extension to infinite dimensions requires measure-theoretic care and may involve subtleties with unbounded operators.
    
    \item \textbf{Axiom strength}: Purification is strong---it excludes classical probability theory, for instance. Are there weaker conditions that still enable the argument?
    
    \item \textbf{Alternative compositions}: We assumed local tomography. Other composition rules might allow different conclusions.
    
    \item \textbf{Operational closedness}: We assumed $\Phi$ applies universally. Could context-dependent probability rules evade our argument while preserving no-signaling?
\end{enumerate}

\section{Conclusion}

We have shown that the Born rule $P(\phi|\psi) = |\langle\phi|\psi\rangle|^2$ is the unique probability assignment compatible with relativistic causality in theories with steering. The argument proceeds by:

\begin{enumerate}
    \item Distinguishing geometric transition probabilities $\tau$ (characterizing state-space structure) from predictive probabilities $P$ (determining experimental predictions).
    
    \item Showing that steering---the ability to remotely prepare ensembles with identical average states---exists in GPTs with purification.
    
    \item Proving that any nonlinear relationship $P = \Phi(\tau)$ creates a signaling channel via steering.
    
    \item Concluding that no-signaling forces $\Phi = \text{id}$, i.e., $P = \tau$.
\end{enumerate}

The key insight is that steering acts as an \emph{amplifier}: it converts the mathematical property of nonlinearity into the physical phenomenon of superluminal signaling. The Born rule is thus not an arbitrary mathematical choice but a \emph{causal fixed point}---the unique probability assignment that remains compatible with no-signaling when entanglement and steering are available.

This perspective suggests that the quadratic form of quantum probabilities reflects deep constraints from the interplay of probability, entanglement, and causality. These constraints would apply to any physical theory with analogous structure, suggesting the Born rule is in some sense \emph{inevitable} for theories with quantum-like features.

\end{document}